\documentclass{article}

\usepackage{amssymb}
\usepackage{eucal}
\usepackage{amsmath}
\usepackage{amsthm}

\numberwithin{equation}{section}
\newtheorem{thm}{\bf Theorem}[section]
\newtheorem{lem}[thm]{\bf Lemma}

\theoremstyle{remark}
\newtheorem{rem}{\bf Remark}[section]

\title{Stability of equilibrium states in the Zhukovski case of heavy gyrostat using algebraic methods}
\author{Dan Com\u anescu\\
{\small Department of Mathematics, West University of Timi\c soara}\\
{\small Bd. V. P\^ arvan, No 4, 300223 Timi\c soara, Rom\^ ania}\\
{\small E-mail addresses: comanescu@math.uvt.ro}}
\date{}

\begin{document}

\maketitle

\begin{abstract}
We study the stability of the equilibrium points of a skew product system. We analyze the possibility to construct a Lyapunov function using a set of conserved quantities and solving an algebraic system. We apply the theoretical results to study the stability of an equilibrium state of a heavy gyrostat in the Zhukovski case.
\end{abstract}



\section{Introduction}

A classical problem in mechanics is the problem of the rotation of a heavy rigid body with a fixed point. If the rigid body is also acted by a gyrostatic torque we obtain the problem of the heavy gyrostat with a fixed point. The rotation of a heavy gyrostat is governed by the differential system, see \cite{birtea-casu-comanescu} and \cite{gavrilov},
\begin{equation}\label{sistem general}
\left\{%
\begin{array}{ll}
\dot{\vec{M}}=(\vec{M}+\vec{\mu})\times \mathbb{I}^{-1}\vec{M} +m\vec{\gamma}\times\vec{r}_G \\
\dot{\vec{\gamma}}=\vec{\gamma}\times\mathbb{I}^{-1}\vec{M},\end{array}%
\right.
\end{equation}
where $m$ is the mass of the gyrostat, $\vec{r}_G$ is the vector with the initial point in the fixed point $O$ and the terminal point in the center of gravity $G$, $\mathbb{I}$ is the moment of inertia tensor at the point $O$, $\vec{M}$ is the angular momentum vector, $\vec{\gamma}$ is the direction of the gravitational field and $\vec{\mu}$ is the constant vector of gyrostatic moment. In the case when $G$ coincide with the fixed point $O$ we have $\vec{r}_G=\vec{0}$ and we obtain the well-known Zhukovski case, see \cite{gavrilov}. The mathematical model of the rotation of a heavy gyrostat in the Zhukovski case is
\begin{equation}\label{M}
\left\{%
\begin{array}{ll}
\dot{\vec{M}}=(\vec{M}+\vec{\mu})\times \mathbb{I}^{-1}\vec{M}\\
\dot{\vec{\gamma}}=\vec{\gamma}\times\mathbb{I}^{-1}\vec{M}.\end{array}%
\right.
\end{equation}
This system admits a Hamilton-Poisson formulation, see \cite{birtea-casu-comanescu} and \cite{gavrilov}. The Poisson bracket is a modified "$-$" Kirillov-Kostant-Souriau bracket on the dual Lie algebra $e(3)^*$. In the coordinates $(M_1,M_2,M_3,\gamma_1,\gamma_2,\gamma_3)$ the matrix associated with this Poisson bracket has the expression
$$\Pi_{\vec{\mu}}(\vec{M},\vec{\gamma})=\left[\begin{array}{cc}
\widehat{\vec{M}+\vec{\mu}} & \widehat{\vec{\gamma}} \\
\widehat{\vec{\gamma}} & \mathcal{O}_3
\end{array}\right],$$
where the matrix $\widehat{v}=\left(
            \begin{array}{ccc}
              0 & -r & q \\
              r & 0 & -p \\
              -q & p & 0 \\
            \end{array}
          \right)$ is defined by the vector $\vec{v}$ which has the components $(p,q,r)$.
If "$\,\cdot\,$" is the scalar product, then we have four conserved quantities: the Hamiltonian function $H$, the Casimir functions $C_1$ and $C_2$ and a fourth conserved quantity $F$. These functions are:
$$H=\frac{1}{2}\vec{M}\cdot \mathbb{I}^{-1}\vec{M},\,\,\,C_1=\frac{1}{2}\vec{\gamma}\cdot\vec{\gamma},\,\,\,C_2=(\vec{M}+\vec{\mu})\cdot \vec{\gamma},\,\,\,\text{and}\,\,\,F=\frac{1}{2}(\vec{M}+\vec{\mu})\cdot (\vec{M}+\vec{\mu}).$$

The system of the rotation of a heavy gyrostat in the Zhukovski case \eqref{M} is a skew product system of differential equations of the form
\begin{equation}\label{semi-decoupled}
\left\{
  \begin{array}{ll}
   \dot{y}=g(y) \\
   \dot{z}=h(y,z)
  \end{array}
\right.
\end{equation}
where $g:D_1\subset \mathbb{R}^{m}\rightarrow \mathbb{R}^{m}$ and $h:D_1\times D_2\subset \mathbb{R}^{m}\times \mathbb{R}^{n-m}\rightarrow \mathbb{R}^{n-m}$ are locally Lipschitz functions and $D_1, D_2$ are open sets. By analogy with the case of the rotations of heavy gyrostat in the Zhukovski case, we suppose that we have complete information on the stability for the reduced differential system
\begin{equation}\label{reduced-system}
    \dot{y}=g(y).
\end{equation}
 Assume that we have $q\geq 1$ conserved quantities $F_1=F_1(y),...,F_q=F_q(y)$ for \eqref{reduced-system}. These functions are conserved quantities for \eqref{semi-decoupled}. We suppose that we have also $k-q\geq 1$ conserved quantities $F_{q+1}=F_{q+1}(y,z),..., F_{k}=F_{k}(y,z)$ for the system \eqref{semi-decoupled}.

In Section 2 we study the stability of the equilibrium points of a differential system of the form \eqref{semi-decoupled}. First, we present a result about the spectral stability of equilibrium points of a differential system of the form \eqref{semi-decoupled}. Using the theory described in \cite{comanescu}, we prove a theorem about the stability with respect to the set of conserved quantities $\{F_1,...,F_k\}$ of an equilibrium point of a system of the form \eqref{semi-decoupled}. An equilibrium point is stable with respect to the set of conserved quantities $\{F_1,...,F_k\}$ if and only if we can construct, in a neighborhood of the equilibrium point, a Lyapunov function using the conserved quantities. At the end of this section we study the partial stability of an equilibrium point. We present a theorem about $y$-instability and $z$-instability of an equilibrium point $(y_e,z_e)$ in the case when the equilibrium point belongs to a certain invariant subset.

In Section 3 we present stability results of the uniform rotations of a torque-free gyrostat. These results have been proved in \cite{comanescu}.

In Section 4 we apply the theoretical results proved in Section 2 for studying the stability of the equilibrium states of a heavy gyrostat. We consider the Zhukovski case and suppose that $\vec{\mu}$, the vector of gyrostatic moment, is situated along a principal axis of inertia of the gyrostat. First, we give a list of equilibrium states and their spectral stability. Second, we analyze the stability of the equilibrium states with respect to the set of conserved quantities $\{H,C_1,C_2,F\}$. The  results of spectral stability and stability with respect to the set of conserved quantities are used to decide the Lyapunov stability of the equilibrium states. There exists some equilibrium states which are spectrally stable and they are not stable with respect to the set of conserved quantities. For a part of them we can prove the Lyapunov instability by using the exact expression of some solutions of \eqref{M}. For the complementary part of them we can not decide the Lyapunov stability or instability. At the end of the section we study the $\vec{M}$-stability and $\vec{\gamma}$-stability of the equilibrium states which is an important practical problem.

\section{The stability of equilibria of a skew product system}

 In this section we study the stability of the equilibrium points ofa differential system of the form \eqref{semi-decoupled}.
The spectral stability of an equilibrium point is a necessary condition for Lyapunov stability. According to \cite{marsden-ratiu}, pp. 32, an equilibrium point $(y_e,z_e)$ is spectrally stable if the linearization of $(g,h)$ at $(y_e,z_e)$ has all the eigenvalues with nonpositive real parts. For our differential system \eqref{semi-decoupled} the linearization at $(y_e,z_e)$ is
$$\frac{\partial (g,h)}{\partial (y,z)}(y_e,z_e)=\left(
                \begin{array}{cc}
                  \frac{\partial g}{\partial y}(y_e) & \mathbf{O}_{m,n-m} \\
                   \frac{\partial h}{\partial y}(y_e,z_e) & \frac{\partial h}{\partial z}(y_e,z_e) \\
                \end{array}
              \right).$$
If $A$ is a square matrix we note by $Sp(A)$ the set of eigenvalues of $A$.
\begin{lem} Let $(y_e,z_e)$ be an equilibrium point of \eqref{semi-decoupled}, then we have
 $Sp(\frac{\partial (g,h)}{\partial (y,z)}(y_e,z_e))=Sp(\frac{\partial g}{\partial y}(y_e))\cup Sp(\frac{\partial h}{\partial z}(y_e,z_e))$.
\end{lem}

\begin{proof}
We remind that if $A$ is a square matrix with $m$ columns, $C$ is a square matrix with $n-m$ columns and $B$ is a matrix with $m$ columns and $n-m$ rows then $\det \left(
             \begin{array}{cc}
               A & \mathbf{O}_{m,n-m} \\
               B & C \\
             \end{array}
           \right)=\det A\det C.$
\end{proof}

In the following we study the stability of equilibrium points with respect to a set of conserved quantities. The stability of an equilibrium point with respect to a set of conserved quantities is a sufficient condition for Lyapunov stability. If an equilibrium point is not stable with respect to a set of conserved quantities, then we cannot construct a Lyapunov function by using this set of conserved quantities. We remind some theoretical considerations, from the paper \cite{comanescu}. We consider an open set $D\subset\mathbb{R}^n$ and the locally Lipschitz function $f:D\rightarrow \mathbb{R}^n$ which generates the differential equation
\begin{equation}\label{ecuatie-diferentiala}
    \dot{x}=f(x)
\end{equation}
Let $x_e$ be an equilibrium point. A continuous function $V:D\rightarrow \mathbb{R}$ which satisfies $V(x_e)=0$ and $V(x)>0$ for every $x$ in a neighborhood of $x_e$ and $x\neq x_e$ is called a positive definite function in the equilibrium point $x_e$.
{\it The equilibrium point $x_e$ of \eqref{ecuatie-diferentiala} is stable with respect to the set of conserved quantities $\{F_1,...,F_k\}$ if there exists a continuous function $\Phi:\mathbb{R}^k\rightarrow \mathbb{R}$ such that $x\rightarrow \Phi(F_1,....,F_k)(x)-\Phi(F_1,....,F_k)(x_e)$ is a positive definite function in $x_e$.}
In the conditions of the above definition the function $x\rightarrow \Phi(F_1,....,F_k)(x)-\Phi(F_1,....,F_k)(x_e)$ is a Lyapunov function in the equilibrium point $x_e$ and we have the following results.
\begin{thm}\label{implication-stability}
If the equilibrium point $x_e$ of \eqref{ecuatie-diferentiala} is stable with respect to the set of conserved quantities $\{F_1,...,F_k\}$ then it is stable in the sense of Lyapunov.
\end{thm}

\begin{thm}\label{stability}
Let $x_e$ be an equilibrium point of \eqref{ecuatie-diferentiala} and $\{F_1,...,F_k\}$ a set of conserved quantities. The following statements are equivalent:
\begin{itemize}
\item[(i)] $x_e$ is stable with respect to the set of conserved quantities $\{F_1,...,F_k\}$;
\item [(ii)] $x\rightarrow ||(F_1,...,F_k)(x)-(F_1,...,F_k)(x_e)||$ is a positive definite function in $x_e$;
\item [(iii)] the system $F_1(x)=F_1(x_e),...,F_k(x)=F_k(x_e)$ has no root besides $x_e$ in some neighborhood of $x_e$.
\end{itemize}
\end{thm}

The Theorem \ref{stability} $(iii)$ offer an algebraic method to prove the Lyapunov stability of an equilibrium point. An interesting proof of this implication is presented in \cite{aeyels}. This method was used in \cite{comanescu} in the problem of stability of uniform rotations of a torque-free gyrostat.

In this paper we study a differential system of the form \eqref{semi-decoupled}.

\begin{thm}\label{semi-decoupled-stability-theory}
Let $(y_e,z_e)$ be an equilibrium point of \eqref{semi-decoupled}.
\begin{itemize}
\item [a)]
Suppose that the equilibrium point $y_e$ of the system \eqref{reduced-system} is stable with respect to the set of conserved quantities $\{F_1,...,F_q\}$. The following statements are equivalent:

\item [a.1.)] The equilibrium point $(y_e,z_e)$ of the system \eqref{semi-decoupled} is stable with respect to the set of conserved quantities $\{F_1,...,F_k\}$;
\item [a.2.)] The solution $z_e$ of the algebraic system $F_{q+1}(y_e,z)=F_{q+1}(y_e,z_e),...,F_k(y_e,z)=F_k(y_e,z_e)$, with the unknown $z$, is isolated in the set of all the solutions.
\item [b)] If the equilibrium point $y_e$ of the system \eqref{reduced-system} is Lyapunov unstable then the equilibrium point $(y_e,z_e)$ of the system \eqref{semi-decoupled} is not stable with respect to the set of conserved quantities  $\{F_1,...,F_k\}$.
\end{itemize}
\end{thm}

\begin{proof} $a)$ The implication $a.1.)\Rightarrow a.2.)$ is trivial.

 Suppose that the proposition $a.2.)$ is true and $(y_e,z_e)$ is not stable with respect to the set of conserved quantities $\{F_1,...,F_k\}$. By using the Theorem \ref{stability} there exists a sequence $(y_n,z_n)_{n\in\mathbb{N}}$ of solutions of the algebraic system $F_1(y,z)=F_1(y_e,z_e),...,F_k(y,z)=F_k(y_e,z_e)$ such that
$(y_n,z_n)\rightarrow (y_e,z_e)$ and $(y_n,z_n)\neq (y_e,z_e)$. All the terms of the sequence $(y_n)_{n\in\mathbb{N}}$ are solutions of the algebraic system $F_1(y)=F_1(y_e),...,F_q(y)=F_q(y_e)$. Because the equilibrium point $y_e$ of the system \eqref{reduced-system} is stable with respect to the set of conserved quantities $\{F_1,...,F_q\}$ we deduce the existence of an index $n_1\in \mathbb{N}$ such that $y_n=y_e$ if $n\geq n_1$. We have that $z_n$, with $n>n_1$, are solutions of the algebraic system $F_{q+1}(y_e,z)=F_{q+1}(y_e,z_e),...,F_k(y_e,z)=F_k(y_e,z_e)$ and $z_n\neq z_e$. We have a contradiction and consequently we obtain the enounced result.

b) By using the definition of the Lyapunov stability it is easy to see the following result: "if the equilibrium point $y_e$ of the system \eqref{reduced-system} is Lyapunov unstable then, the equilibrium point $(y_e,z_e)$ of the system \eqref{semi-decoupled} is Lyapunov unstable". Suppose that the equilibrium point $(y_e,z_e)$ of the system \eqref{semi-decoupled} is stable with respect to the set of conserved quantities $\{F_1,...,F_k\}$ and we deduce, by using the Theorem \ref{implication-stability}, the Lyapunov stability of $(y_e,z_e)$. We have a contradiction and consequently, the equilibrium point $(y_e,z_e)$ of the system \eqref{semi-decoupled} is not stable with respect to the set of conserved quantities  $\{F_1,...,F_k\}$.
\end{proof}

\begin{rem}
The above theorem does not offer information about the stability with respect to the set of conserved quantities $F_1,...,F_k$ of an equilibrium point $(y_e,z_e)$ in the case when $y_e$ is a Lyapunov stable equilibrium of \eqref{reduced-system} but it is not stable with respect to the set of conserved quantities $F_1,...,F_q$.
\end{rem}

In practical problems, the study of stability of an equilibrium point of a differential system \eqref{semi-decoupled} implies the study of partial stability with respect to $y$ and $z$. According to \cite{rouche}, pp. 15, the equilibrium point $(y_e,z_e)$ is $y$-stable (respectively $z$-stable) if for all $\varepsilon>0$ there exists $\delta>0$ such that $||y(t, y_0,z_0)-y_e||<\varepsilon $ (respectively $||z(t, y_0,z_0)-z_e||<\varepsilon$) for all initial conditions $||(y_0,z_0)-(y_e,z_e)||<\delta$ and all $t\geq 0$. About the partial stability of the equilibrium points of a skew product system we have the following properties.

\begin{thm}\label{partial-semi-decoupled-generality}
Let $(y_e,z_e)$ be an equilibrium point of \eqref{semi-decoupled}.
\begin{itemize}
\item [(i)] If $(y_e,z_e)$ is Lyapunov stable, then it is $y$-stable and $z$-stable.
\item [(ii)] The equilibrium point $y_e$ of \eqref{reduced-system} is Lyapunov stable if and only if $(y_e,z_e)$ is $y$-stable.
\item [(iii)] If the equilibrium point $y_e$ of \eqref{reduced-system} is Lyapunov stable, then $(y_e,z_e)$ is Lyapunov stable if and only if it is $z$-stable.
\end{itemize}

\end{thm}

The above theorem does not offer sufficient information about $z$-stability of an equilibrium points. We present a method to prove $z$-instability in a particular case which will appear in our concrete problem. This method assumes the existence of a certain invariant set of the dynamics generated by \eqref{semi-decoupled}. The invariant set can by found by the method presented in the paper \cite{birtea-comanescu}.
For the following considerations we suppose that $n=2 m$ and there exists the $C^1$-function $s:\mathbb{R}^{2m}\rightarrow \mathbb{R}^m$ such that the set
\begin{equation}\label{}
    \mathcal{M}=\{(y,z)\,|\,s(y,z)=0\}
\end{equation}
is invariant under the dynamics generated by \eqref{semi-decoupled}.

\begin{thm}\label{instability}
Let $(y_e,z_e)\in \mathcal{M}$ be an equilibrium point of \eqref{semi-decoupled}. Suppose that the matrices $\frac{\partial s}{\partial y}(y_e,z_e)$ and $\frac{\partial s}{\partial z}(y_e,z_e)$ are invertible. If $y_e$ is not Lyapunov stable for the dynamics generated by \eqref{reduced-system} then $(y_e,z_e)$ is not $y$-stable and it is not $z$-stable.
\end{thm}

\begin{proof}
The fact that $(y_e,z_e)$ of  is not $y$-stable is a consequence of the Theorem \ref{partial-semi-decoupled-generality} $(ii)$.

By using the implicit function theorem we have an open set $U\subset \mathbb{R}^m$ containing $y_e$, an open set $V\subset \mathbb{R}^m$ containing $z_e$ and $C^1$ function $r:U\rightarrow V$ such that $r(y_e)=z_e$ and $\{(y,r(y))\,|\,y\in U\}=\{(y,z)\in U\times V\,|\,s(y,z)=0\}$. If $y\in U$, then $s(y,r(y))=0$ and we deduce that
$\frac{\partial r}{\partial y}(y_e)=-[\frac{\partial s}{\partial z}(y_e,z_e)]^{-1}\frac{\partial s}{\partial y}(y_e,z_e)$. By hypotheses we have that $\frac{\partial r}{\partial y}(y_e)$ is an invertible matrix. By local inversion theorem we have that exists an open set $W\subset \mathbb{R}^m$ containing $y_e$ such that the function $(y\in W)\rightarrow ||r(y)-r(y_e)||$ is a positive definite function and we deduce the existence of a constant $k>0$ such that $||r(y)-r(y_e)||\geq k||y-y_e||$ for $y\in W$. By hypotheses $y_e$ is not Lyapunov stable for the dynamics generated by \eqref{reduced-system}. There exists $\varepsilon>0$, a sequence $y_i\rightarrow y_e$ and a sequence $t_i>0$ such that $B(y_e,\frac{\varepsilon}{k})\subset U$ and  $||y(t_i,y_i)-y_e||=\frac{\varepsilon}{k}$. We use the notation $y(\cdot,y_0)$ for the solution of \eqref{reduced-system} which start from $y_0$ and an analogous notation for the solutions of \eqref{semi-decoupled}. We denote by $z_i=r(y_i)$ and by continuity we $z_i\rightarrow z_e$. Because $(y_i,z_i)\in \mathcal{M}$ and $\mathcal{M}$ is invariant under the dynamics generated by \eqref{semi-decoupled} we deduce that $(y(t,y_i,z_i),z(t,y_i,z_i))\in \mathcal{M}$ for all $t$. With the above notations we have
$$||z(t_i,y_i,z_i)-z_e||=||r(y(t_i,y_i,z_i))-r(y_e)||>k||y(t_i,y_i)-y_e||=\varepsilon,$$
which implies the $z$-instability of the equilibrium point $(y_e,z_e)$.
\end{proof}

\section{The stability of the uniform rotations of a torque-free gyrostat}\label{torque-free gyrostat}

The problem of the rotations of a torque-free gyrostat is a classical problem in the theory of the motion of a rigid body (see \cite{wittenburg}). The study of stability of uniform rotations appears in many papers. In \cite{puta-comanescu} is studied an equivalent differential system and the results of this paper are used in \cite{comanescu} to study the stability with respect to a set of conserved quantities and Lyapunov stability of uniform rotations.
The torque-free gyrostat equation is
\begin{equation}\label{torque-free}
    \dot{\vec{M}}=(\vec{M}+\vec{\mu})\times \mathbb{I}^{-1}\vec{M}
\end{equation}
This equation has two conserved quantities
$H=\frac{1}{2}\vec{M}\cdot \mathbb{I}^{-1}\vec{M}\,\,\text{and}\,\,F=\frac{1}{2}(\vec{M}+\vec{\mu})\cdot (\vec{M}+\vec{\mu}).$
We observe that if we replace the vector of gyrostatic moment $\vec{\mu}$ with $-\vec{\mu}$ then we obtain the equation
\begin{equation}\label{minus-torque-free}
    \dot{\vec{M}}=(\vec{M}-\vec{\mu})\times \mathbb{I}^{-1}\vec{M}.
\end{equation}
A function $\vec{M}:\mathbb{R}\rightarrow \mathbb{R}^3$ is a solution of \eqref{torque-free} if and only if the function $-\vec{M}$ is a solution of \eqref{minus-torque-free}. We have the same stability properties for an uniform rotations $\vec{M}_e$ of \eqref{torque-free} and for an uniform rotations $-\vec{M}_e$ of \eqref{minus-torque-free}.

We denote by $I_1,I_2$ and $I_3$ the principal moment of inertia and suppose that $I_1>I_2>I_3$. For the following considerations we use a body frame for which the axes are principal axes of inertia. The matrix of the moment of inertia tensor in this body frame has the form
$\mathbb{I}=\hbox{diag}(I_1,I_2,I_3)$. Also, we denote by $\mu_1,\mu_2$ and $\mu_3$ the components of the vector $\vec{\mu}$ with respect to this frame.
We remind, see \cite{comanescu}, some results of stability for the uniform rotations in the cases in which the gyrostatic moment is directed in the positive sense along a principal axis of inertia. Using the above observations we have analogous results for the cases in which the gyrostatic moment is directed in the negative sense along a principal axis of inertia.

{\bf I.} An uniform rotation is Lyapunov stable if and only if it is stable with respect to the set of conserved quantities $\{H,F\}$.

{\bf II.} For the case $\mu_1>0,\,\mu_2=0,\,\mu_3=0$ we have:
\begin{itemize}
\item [(1)] An uniform rotation of the form $\vec{M}_{1-2}=(q,0,0)$, with $q\in \mathbb{R}$, has the properties:
\begin{itemize}
\item [(1.1)] It is spectrally stable if and only if $q\in (-\infty, -\frac{I_1\mu_1}{I_1-I_2}]\cup [-\frac{I_1\mu_1}{I_1-I_3},\infty)$.
\item [(1.2)] It is stable with respect to the set of conserved quantities $\{H,F\}$ if and only if $q\in (-\infty, -\frac{I_1\mu_1}{I_1-I_2})\cup [-\frac{I_1\mu_1}{I_1-I_3},\infty)$.
\end{itemize}
\item [(2)] An uniform rotation of the form $\vec{M}_{4}=(\frac{I_1}{I_2-I_1}\mu_1, q,0)$, $q\neq 0$, is spectrally unstable.
\item [(3)] An uniform rotation of the form $\vec{M}_{5}=(\frac{I_1}{I_3-I_1}\mu_1, 0, q)$, $q\neq 0$, is stable with respect to the set of conserved quantities $\{H,F\}$.
\end{itemize}

{\bf III.} For the case $\mu_1=0,\mu_2>0,\,\mu_3=0$ we have:
\begin{itemize}
\item [(1)] An uniform rotation of the form $\vec{M}_{1-2}=(0,q,0)$ is stable with respect to the set of conserved quantities $\{H,F\}$ if and only if $q\in [-\frac{I_2\mu_2}{I_2-I_3},\frac{I_2\mu_2}{I_1-I_2}]$.
\item [(2)] An uniform rotation of the form $\vec{M}_{3}=(q, \frac{I_2}{I_1-I_2}\mu_2, 0)$ or $\vec{M}_{5}=(0, \frac{I_2}{I_3-I_2}\mu_2, q)$, $q\neq 0$, is stable with respect to the set of conserved quantities $\{H,F\}$.
\end{itemize}

{\bf IV.} For the case $\mu_1=0,\,\mu_2=0,\,\mu_3>0$ we have:
\begin{itemize}
\item [(1)] An uniform rotation of the form $\vec{M}_{1-2}=(0,0,q)$, with $q\in \mathbb{R}$, has the properties:
\begin{itemize}
\item [(1.1)] It is spectrally stable if and only if $q\in (-\infty, \frac{I_3\mu_3}{I_1-I_3}]\cup [\frac{I_3\mu_3}{I_2-I_3},\infty)$.
\item [(1.2)] It is stable with respect to the set of conserved quantities $\{H,F\}$ if and only if $q\in (-\infty, \frac{I_3\mu_3}{I_1-I_3}]\cup (\frac{I_3\mu_3}{I_2-I_3},\infty)$.
\end{itemize}
\item [(2)] An uniform rotation of type $\vec{M}_3=(q,0,\frac{I_3}{I_1-I_3}\mu_3)$, $q\neq 0$, is stable with respect to the set of conserved quantities $\{H,F\}$.
\item [(3)] An uniform rotation of type $\vec{M}_4=(0,q,\frac{I_3}{I_2-I_3}\mu_3)$, $q\neq 0$, is spectrally unstable.
\end{itemize}

\section{The stability of the equilibrium states in the Zhukovski case of integrability of a heavy gyrostat}

In this section we apply the results from Section 2 for the Zhukovski case of a heavy gyrostat. We consider the cases in which the vector of gyrostatic moment is situated along a principal axis of inertia of the gyrostat.
We denote, also, by $I_1,I_2$ and $I_3$ the principal moment of inertia and suppose that $I_1>I_2>I_3$. The matrix of the moment of inertia tensor in our body frame has the form
$\mathbb{I}=\hbox{diag}(I_1,I_2,I_3)$ and $\mu_1,\mu_2,\mu_3$ are the components of the vector $\vec{\mu}$.

\subsection{The equilibrium states}

The equations of the equilibrium states are
$$(\vec{M}+\vec{\mu})\times \mathbb{I}^{-1}\vec{M}=\vec{0},\,\,\,\vec{\gamma}\times \mathbb{I}^{-1}\vec{M}=\vec{0}.$$
It is easy to see that we have the following result.
\begin{lem}\label{echilibre-M}  The equilibrium states of \eqref{M} are characterized by the following properties.
\begin{itemize}
\item [(i)] If $(\vec{M}_e,\vec{\gamma}_e)$ is an equilibrium state for the system \eqref{M} then $\vec{M}_e$ is an equilibrium state for \eqref{torque-free}.
\item [(ii)] For all $\vec{\gamma}_e\in \mathbb{R}^3$, the states $(\vec{0},\vec{\gamma}_e)$ are equilibrium states for \eqref{M}.
\item [(iii)] If $\vec{M}_e\neq \vec{0}$ is an equilibrium state for \eqref{torque-free} then $(\vec{M}_e,\vec{\gamma}_e)$ is an equilibrium state for the system \eqref{M} if and only if there exists $\theta \in \mathbb{R}$ such that $\vec{\gamma}_e=\theta\mathbb{I}^{-1}\vec{M}_e$.
\end{itemize}
\end{lem}

According to the results of Section \ref{torque-free gyrostat} and using Lemma \eqref{echilibre-M} we obtain:
\begin{itemize}
\item [(I)] If $\mu_2=\mu_3=0$ we have the equilibrium states:
$$(\vec{M},\vec{\gamma})_0=(0,0,0,\alpha_1,\alpha_2,\alpha_3),\,\,(\vec{M},\vec{\gamma})_{1-2}=(q,0,0,\alpha,0,0),$$
$$(\vec{M},\vec{\gamma})_4=(\frac{I_1}{I_2-I_1}\mu_1, \beta,0,\frac{\theta\mu_1}{I_2-I_1},\frac{\theta\beta}{I_2},0),\,\,(\vec{M},\vec{\gamma})_5=(\frac{I_1}{I_3-I_1}\mu_1, 0, \beta,\frac{\theta\mu_1}{I_3-I_1},0,
\frac{\theta\beta}{I_3}).$$
\item [(II)] If $\mu_1=\mu_3=0$ we have the equilibrium states:
$$(\vec{M},\vec{\gamma})_0=(0,0,0,\alpha_1,\alpha_2,\alpha_3),\,\,(\vec{M},\vec{\gamma})_{1-2}=(0,q,0,0,\alpha,0),$$
$$(\vec{M},\vec{\gamma})_3=(\beta, \frac{I_2}{I_1-I_2}\mu_2, 0,\frac{\theta\beta}{I_1},\frac{\theta\mu_2}{I_1-I_2},
0),\,\,\,(\vec{M},\vec{\gamma})_5=(0, \frac{I_2}{I_3-I_2}\mu_2, \beta,0,\frac{\theta\mu_2}{I_3-I_2},
\frac{\theta\beta}{I_3}).$$
\item [(III)] If $\mu_1=\mu_2=0$ we have the equilibrium states:
$$(\vec{M},\vec{\gamma})_0=(0,0,0,\alpha_1,\alpha_2,\alpha_3),\,\,(\vec{M},\vec{\gamma})_{1-2}=(0,0,q,0,0,\alpha),$$
$$(\vec{M},\vec{\gamma})_3=(\beta, 0, \frac{I_3}{I_1-I_3}\mu_3,\frac{\theta\beta}{I_1},0,
\frac{\theta\mu_3}{I_1-I_3}),\,\,\,(\vec{M},\vec{\gamma})_4=(0, \beta, \frac{I_3}{I_2-I_3}\mu_3,0,\frac{\theta\beta}{I_2},
\frac{\theta\mu_3}{I_2-I_3}).$$
\end{itemize}
In the above expressions we have $\alpha,\alpha_1,\alpha_2,\alpha_3,\theta\in \mathbb{R}$ and $q,\beta\in\mathbb{R}^*$.

\subsection{Spectral stability of the equilibrium states}
We prove that spectral stability of the equilibrium states is characterized by spectral stability of the uniform rotations of a torque-free gyrostat. Spectral stability of an uniform rotation of the torque-free gyrostat is studied in paper \cite{comanescu} using some computations which appear in paper \cite{puta-comanescu}. These results are presented in Section \ref{torque-free gyrostat}.

\begin{thm}\label{spectral-stability}
An equilibrium state $(\vec{M}_e,\vec{\gamma}_e)$ of the system \eqref{M} is spectrally stable if and only if the uniform rotation $\vec{M}_e$ of the torque-free gyrostat equation \eqref{torque-free} is spectrally stable.
\end{thm}

\begin{proof}
The linearization of the right side of the equation \eqref{M} is
$$\mathcal{L}(\vec{M},\vec{\gamma})=\left(
                                                   \begin{array}{cc}
                                                     \mathcal{L}_M(\vec{M}) & \mathcal{O}_3 \\
                                                     \widehat{\vec{\gamma}}\cdot\mathbb{I}^{-1} & \widehat{\mathbb{I}^{-1}\vec{M}} \\
                                                   \end{array}
                                                 \right),$$
where $\mathcal{L}_M(\vec{M})$ is the linearization of the right side of \eqref{torque-free}. The characteristic polynomial of $\mathcal{L}(\vec{M}_e,\vec{\gamma}_e)$ is
$$\mathcal{P}_{(\vec{M}_e,\vec{\gamma}_e)}(t)=-t(t^2+||\mathbb{I}^{-1}\vec{M}_e||^2)\mathcal{P}_{\vec{M}_e}(t),$$
where $\mathcal{P}_{\vec{M}_e}(t)$ is the characteristic polynomial of $\mathcal{L}_M(\vec{M}_e)$.
\end{proof}

\subsection{Stability with respect to the set of conserved quantities}

In this section we study the stability of an equilibrium state of the system \eqref{M} with respect to the set $\mathcal{CQ}=\{H,C_1,C_2,F\}$ of conserved quantities. To apply Theorem \ref{semi-decoupled-stability-theory} it is necessary to study the algebraic system
\begin{equation}\label{reduced}
    \left\{
      \begin{array}{ll}
        C_1(\vec{\gamma})= C_1(\vec{\gamma}_0)\\
        C_2(\vec{\gamma})= C_2(\vec{\gamma}_0)
      \end{array}
    \right.
\Leftrightarrow
\left\{
      \begin{array}{ll}
        \vec{\gamma}\cdot\vec{\gamma}= \vec{\gamma}_0\cdot\vec{\gamma}_0\\
        (\vec{M}_0+\vec{\mu})\cdot\vec{\gamma}= (\vec{M}_0+\vec{\mu})\cdot\vec{\gamma}_0
      \end{array}
    \right.,
\end{equation}
where $\vec{M}_0,\vec{\gamma}_0\in \mathbb{R}^3$ and $\vec{\gamma}$ is the unknown.

\begin{lem}\label{solve-system-reduced}
The solutions of the system \eqref{reduced} have the properties:
\begin{itemize}
\item [(i)] if $\vec{\gamma}_0=\vec{0}$, then $\vec{\gamma}_0$ is the unique solution of the system \eqref{reduced};
\item [(ii)] if $\vec{M}_0=-\vec{\mu}$ and $\vec{\gamma}_0\neq\vec{0}$, then the set of solutions of \eqref{reduced} is $\{\vec{\gamma}\,|\,\vec{\gamma}\cdot\vec{\gamma}= \vec{\gamma}_0\cdot\vec{\gamma}_0\}$;
\item [(iii)] if $\vec{M}_0\neq -\vec{\mu}$ and $\vec{\gamma}_0\neq\vec{0}$, then the solution $\vec{\gamma}_0$ is isolated in the set of all the solutions of the system \eqref{reduced} if and only if the vectors $\vec{M}_0+\vec{\mu}$ and $\vec{\gamma}_0$ are linearly dependent.
\end{itemize}
\end{lem}

\begin{proof}
It is easy to see that the affirmations $(i)$ and $(ii)$ are true.

In the hypotheses of $(iii)$ the solutions are the intersection of a sphere with a plane. The sphere and the plane have a common point $\vec{\gamma}_0$. In the case in which the plane is tangent to the sphere the intersection is formed by the unique point $\vec{\gamma}_0$. In the case in which the plane is not tangent to the sphere, then the intersection is a circle and $\vec{\gamma}_0$ is not an isolated point in the set of all the solutions. We observe that the plane is tangent to the sphere if and only if the vectors $\vec{M}_0+\vec{\mu}$ and $\vec{\gamma}_0$ are linearly dependent and we deduce the enounced result.
\end{proof}

\begin{thm}\label{stability-conserved-quantities-first-axis}
 In the case of gyrostatic moment along the first axis of inertia, $\mu_1>0,\,\mu_2=0,\,\mu_3=0$, we have the following results about the equilibrium states.
\begin{itemize}
\item [a)] $(\vec{M},\vec{\gamma})_0=(0,0,0,\alpha_1,\alpha_2,\alpha_3)$ is stable with respect to the set of conserved quantities $\mathcal{CQ}$ if and only if $\alpha_2=\alpha_3=0$.
\item [b)] $(\vec{M},\vec{\gamma})_{1-2}=(q,0,0,\alpha,0,0)$ is stable with respect to the set of conserved quantities $\mathcal{CQ}$ if and only if  $q\in (-\infty, -\frac{I_1\mu_1}{I_1-I_2})\cup [-\frac{I_1\mu_1}{I_1-I_3},\infty)$ and $q\neq -\mu_1$ or $q=-\mu_1$ and $\alpha=0$.
\item [c)] $(\vec{M},\vec{\gamma})_4=(\frac{I_1}{I_2-I_1}\mu_1, \beta,0,\frac{\theta\mu_1}{I_2-I_1},\frac{\theta\beta}{I_2},0)$ is not stable with respect to the set of conserved quantities $\mathcal{CQ}$.
\item [d)] $(\vec{M},\vec{\gamma})_5=(\frac{I_1}{I_3-I_1}\mu_1, 0, \beta,\frac{\theta\mu_1}{I_3-I_1},0,
\frac{\theta\beta}{I_3})$ is stable with respect to the set of conserved quantities $\mathcal{CQ}$.
\end{itemize}
\end{thm}

\begin{proof}
$a)$ Using the results of Section \ref{torque-free gyrostat} we have that $(0,0,0)$ is an uniform rotation of the torque-free gyrostat system \eqref{torque-free} which is stable with respect to the set of conserved quantities $\{H,F\}$. We can apply Theorem \ref{semi-decoupled-stability-theory} (a.) and consequently we have that the equilibrium state $(\vec{M},\vec{\gamma})_0$ is stable with respect to the set of conserved quantities $\mathcal{CQ}$ if and only if the vector $(\alpha_1,\alpha_2,\alpha_3)$ is an isolated solution of the algebraic system \eqref{reduced}, with $\vec{M}_0=(0,0,0)$ and $\vec{\gamma}_0=(\alpha_1,\alpha_2,\alpha_3)$. By using Lemma \ref{solve-system-reduced} $(i)$ and $(iii)$ we have that $(\vec{M},\vec{\gamma})_0$ is stable with respect to the set of conserved quantities $\mathcal{CQ}$ if and only if the vectors $(\mu_1,0,0)$ and $(\alpha_1,\alpha_2,\alpha_3)$ are linearly dependent and consequently, $\alpha_2=\alpha_3=0$.

b) If $q\notin (-\infty, -\frac{I_1\mu_1}{I_1-I_2})\cup [-\frac{I_1\mu_1}{I_1-I_3},\infty)$, then from Section \ref{torque-free gyrostat} we have that the uniform rotation $(q,0,0)$ of the torque-free gyrostat system is Lyapunov unstable. We apply Theorem \ref{semi-decoupled-stability-theory} (b.) and consequently, the equilibrium state $(\vec{M},\vec{\gamma})_{1-2}$ is not stable with respect to the set of conserved quantities $\mathcal{CQ}$.

If $q=-\mu_1$, then the uniform rotation $(q,0,0)$ of the torque-free gyrostat system is stable with respect to the set of conserved quantities $\{H,F\}$. If $\alpha=0$, then by Lemma \ref{solve-system-reduced} $(i)$ the algebraic system \eqref{reduced}, with $\vec{M}_0=(-\mu_1,0,0)$ and $\vec{\gamma}_0=(0,0,0)$, has a unique solution and we deduce that the equilibrium state $(-\mu_1,0,0,0,0,0)$ is stable with respect to the set of conserved quantities $\mathcal{CQ}$. If $\alpha\neq 0$, then by Lemma \ref{solve-system-reduced} $(ii)$ the solution $(\alpha,0,0)$ of the algebraic system \eqref{reduced}, with $\vec{M}_0=(-\mu_1,0,0)$ and $\vec{\gamma}_0=(\alpha,0,0)$, is not isolated in the set of all the solutions and we deduce that the equilibrium state $(-\mu_1,0,0,\alpha,0,0)$ is not stable with respect to the set of conserved quantities $\mathcal{CQ}$.

If $q\in (-\infty, -\frac{I_1\mu_1}{I_1-I_2})\cup [-\frac{I_1\mu_1}{I_1-I_3},\infty)$ and $q\neq -\mu_1$, then from Section \ref{torque-free gyrostat} we have that the uniform rotation $(q,0,0)$ of the torque-free gyrostat system is stable with respect to the set of conserved quantities $\{H,F\}$. We observe that the vectors $\vec{M}_{1-2}+\vec{\mu}=(q+\mu_1,0,0)$ and $\vec{\gamma}_{1-2}=(\alpha,0,0)$ are linearly dependent. By using Lemma \ref{solve-system-reduced} and Theorem \ref{semi-decoupled-stability-theory} (a.) we deduce that the equilibrium state $(\vec{M},\vec{\gamma})_{1-2}$ is stable with respect to the set of conserved quantities $\mathcal{CQ}$.

c) From Section \ref{torque-free gyrostat} we know that the uniform rotation $(\frac{I_1}{I_2-I_1}\mu_1, \beta,0)$ of the torque-free gyrostat system is Lyapunov unstable. We apply Theorem \ref{semi-decoupled-stability-theory} (b.) and consequently, the equilibrium state $(\vec{M},\vec{\gamma})_{4}$ is not stable with respect to the set of conserved quantities $\mathcal{CQ}$.

d) The uniform rotation $(\frac{I_1}{I_3-I_1}\mu_1, 0, \beta)$ of the torque-free gyrostat system is stable with respect to the set of conserved quantities $\{H,F\}$. We observe that the vectors $\vec{M}_{5}+\vec{\mu}=(\frac{I_1}{I_3-I_1}\mu_1+\mu_1,0,\beta)$ and $\vec{\gamma}_{5}=(\frac{\theta\mu_1}{I_3-I_1},0,
\frac{\theta\beta}{I_3})$ are linearly dependent. By using Lemma \ref{solve-system-reduced} and Theorem \ref{semi-decoupled-stability-theory} (a.) we deduce that the equilibrium state $(\vec{M},\vec{\gamma})_{5}$ is stable with respect to the set of conserved quantities $\mathcal{CQ}$.
\end{proof}

Using analogous considerations we can prove the following results.

\begin{thm} In the case of gyrostatic moment along the second axis of inertia, $\mu_1=0,\,\mu_2>0,\,\mu_3=0$, we have the following results about the equilibrium states.
\begin{itemize}
\item [a)] $(\vec{M},\vec{\gamma})_0=(0,0,0,\alpha_1,\alpha_2,\alpha_3)$ is stable with respect to the set of conserved quantities $\mathcal{CQ}$ if and only if $\alpha_1=\alpha_3=0$.
\item [b)] $(\vec{M},\vec{\gamma})_{1-2}=(0,q,0,0,\alpha,0)$ is stable with respect to the set of conserved quantities $\mathcal{CQ}$ if and only if  $q\in [-\frac{I_2\mu_2}{I_2-I_3},\frac{I_2\mu_2}{I_1-I_2}]$ and $q\neq -\mu_2$ or $q=-\mu_2$ and $\alpha=0$.
\item [c)] $(\vec{M},\vec{\gamma})_3=(\beta, \frac{I_2}{I_1-I_2}\mu_2, 0,\frac{\theta\beta}{I_1},\frac{\theta\mu_2}{I_1-I_2},
0)$ and $(\vec{M},\vec{\gamma})_5=(0, \frac{I_2}{I_3-I_2}\mu_2, \beta,0,\frac{\theta\mu_2}{I_3-I_2},
\frac{\theta\beta}{I_3})$ are stable with respect to the set of conserved quantities $\mathcal{CQ}$.

\end{itemize}
\end{thm}

\begin{thm} In the case of gyrostatic moment along the third axis of inertia, $\mu_1=0,\,\mu_2=0,\,\mu_3>0$, we have the following results about the equilibrium states.
\begin{itemize}
\item [a)] $(\vec{M},\vec{\gamma})_0=(0,0,0,\alpha_1,\alpha_2,\alpha_3)$ is stable with respect to the set of conserved quantities $\mathcal{CQ}$ if and only if $\alpha_1=\alpha_2=0$.
\item [b)] $(\vec{M},\vec{\gamma})_{1-2}=(0,0,q,0,0,\alpha)$ is stable with respect to the set of conserved quantities $\mathcal{CQ}$ if and only if  $q\in (-\infty, \frac{I_3\mu_3}{I_1-I_3}]\cup (\frac{I_3\mu_3}{I_2-I_3},\infty)$ and $q\neq -\mu_3$ or $q=-\mu_3$ and $\alpha=0$.
\item [c)] $(\vec{M},\vec{\gamma})_3=(\beta, 0, \frac{I_3}{I_1-I_3}\mu_3,\frac{\theta\beta}{I_1},0,
\frac{\theta\mu_3}{I_1-I_3})$ is stable with respect to the set of conserved quantities $\mathcal{CQ}$.
\item [d)] $(\vec{M},\vec{\gamma})_4=(0, \beta, \frac{I_3}{I_2-I_3}\mu_3,0,\frac{\theta\beta}{I_2},
\frac{\theta\mu_3}{I_2-I_3})$ is not stable with respect to the set of conserved quantities $\mathcal{CQ}$.
\end{itemize}
\end{thm}

\subsection{The Lyapunov stability of the equilibrium states}

In this section we study the Lyapunov stability of the equilibrium states. To prove the Lyapunov stability of an equilibrium state we use its stability with respect to the set of conserved quantities $\mathcal{CQ}=\{H,C_1,C_2,F\}$. To prove that an equilibrium state is not Lyapunov stable we use its spectrally instability or we apply Theorem \ref{instability}. There exists some cases of equilibrium points for which we can not decide the Lyapunov stability.

\begin{thm}\label{Lyapunov}
Let $(\vec{M}_e,\vec{\gamma}_e)$ be an equilibrium state of \eqref{M}.
\begin{itemize}
\item [(i)] In the case of a gyrostatic moment along the first axis of inertia, $\mu_1>0,\,\mu_2=0,\,\mu_3=0$, and if $(\vec{M}_e,\vec{\gamma}_e)\notin \{(-\mu_1,0,0,\alpha,0,0)\,|\, \alpha\neq 0\}$, then the equilibrium state $(\vec{M}_e,\vec{\gamma}_e)$ is Lyapunov stable if and only if it is stable with respect to the set of conserved quantities $\mathcal{CQ}$.
\item [(ii)] In the case of a gyrostatic moment along the second axis of inertia, $\mu_1=0,\,\mu_2>0,\,\mu_3=0$, and if $(\vec{M}_e,\vec{\gamma}_e)\notin \{(0,-\mu_2,0,0,\alpha,0)\,|\, \alpha\neq 0\}$, then the equilibrium state $(\vec{M}_e,\vec{\gamma}_e)$ is Lyapunov stable if and only if it is stable with respect to the set of conserved quantities $\mathcal{CQ}$.
\item [(iii)] In the case of a gyrostatic moment along the third axis of inertia, $\mu_1=0,\,\mu_2=0,\,\mu_3>0$, and if $(\vec{M}_e,\vec{\gamma}_e)\notin \{(0,0,-\mu_3,0,0,\alpha)\,|\, \alpha\neq 0\}$, then the equilibrium state $(\vec{M}_e,\vec{\gamma}_e)$ is Lyapunov stable if and only if it is stable with respect to the set of conserved quantities $\mathcal{CQ}$.

\end{itemize}
\end{thm}

\begin{proof}
$(i)$
An equilibrium point of the form $(\vec{M},\vec{\gamma})_0=(0,0,0,\alpha_1,\alpha_2,\alpha_3)$, with $\alpha_2=\alpha_3=0$, is Lyapunov stable because it is stable with respect to the set of conserved quantities $\mathcal{CQ}$ (see Theorems \ref{stability-conserved-quantities-first-axis} and \ref{implication-stability}).

We prove that if $\alpha_2\neq 0$ or $\alpha_3\neq 0$, then the equilibrium state $(\vec{M},\vec{\gamma})_0$ is not Lyapunov stable. We study the solution of the system \eqref{M} which verify the initial condition $(q,0,0,\alpha_1,\alpha_2,\alpha_3)$. It is easy to see that $M_1(t)=q,\,M_2(t)=M_3(t)=0$. The components $\gamma_1,\,\gamma_2$ and $\gamma_3$ are the solution of the Cauchy problem
$$\dot{\gamma}_1=0,\,\,\dot{\gamma}_2=\frac{q}{I_1}\gamma_3,\,\,\dot{\gamma}_3=-\frac{q}{I_1}\gamma_2,\,\,\gamma_1(0)=\alpha_1,\,\,\gamma_2(0)=\alpha_2,\,\,\gamma_3(0)=\alpha_3.$$
We obtain
$$\gamma_1(t)=\alpha_1,\,\,\gamma_2(t)=\alpha_2\cos(\frac{q}{I_1}t)+\alpha_3\sin(\frac{q}{I_1}t),\,\,\gamma_3(t)=\alpha_3\cos(\frac{q}{I_1}t)-\alpha_2\sin(\frac{q}{I_1}t).$$
We observe that for all $q\in \mathbb{R}^*$ we have $[-|\alpha_2|,|\alpha_2|]\subset \gamma_2(\mathbb{R})$ and $[-|\alpha_3|,|\alpha_3|]\subset \gamma_3(\mathbb{R})$ which implies that the equilibrium state $(\vec{M},\vec{\gamma})_0$ is not Lyapunov stable.

If $q\in (-\infty, -\frac{I_1\mu_1}{I_1-I_2})\cup [-\frac{I_1\mu_1}{I_1-I_3},\infty)$ and $q\neq -\mu_1$ or $q=-\mu_1$ and $\alpha=0$, then by using Theorems \ref{stability-conserved-quantities-first-axis} and \ref{implication-stability} we deduce that the equilibrium state $(\vec{M},\vec{\gamma})_{1-2}=(q,0,0,\alpha,0,0)$ is Lyapunov stable.

If $q\in (-\frac{I_1\mu_1}{I_1-I_2}, -\frac{I_1\mu_1}{I_1-I_3})$ then the equilibrium state $(\vec{M},\vec{\gamma})_{1-2}=(q,0,0,\alpha,0,0)$ is not spectrally stable and consequently, it is not Lyapunov stable.

If $q=-\frac{I_1\mu_1}{I_1-I_2}$, then by using the results of Section \ref{torque-free gyrostat} we have that the uniform rotation $(q,0,0)$ of the torque-free gyrostat equation \eqref{torque-free} is not Lyapunov stable and, consequently the equilibrium state $(\vec{M},\vec{\gamma})_{1-2}=(-\frac{I_1\mu_1}{I_1-I_2},0,0,\alpha,0,0)$ is not Lyapunov stable for the system \eqref{M}.

According to Theorem \ref{spectral-stability} and the results of Section \ref{torque-free gyrostat} the equilibrium state $(\vec{M},\vec{\gamma})_4$ is not spectrally stable and consequently, it is not Lyapunov stable.

An equilibrium state $(\vec{M},\vec{\gamma})_5$ is Lyapunov stable because it is stable with respect to the set of conserved quantities $\mathcal{CQ}$.
\medskip

Analogously we obtain the results enounced in $(ii)$ and $(iii)$.

\end{proof}

\begin{rem}
In the above theorem and for the case of a gyrostatic moment along the first axis of inertia, $\mu_1>0,\,\mu_2=0,\,\mu_3=0$, the Lyapunov stability of the equilibrium states of the form $(-\mu_1,0,0,\alpha,0,0)$ with $\alpha\neq 0$ cannot be decided. This equilibrium states have the properties:
\begin{itemize}
\item [(i)] They are spectrally stable, see Theorem \ref{spectral-stability} and the results of Section \ref{torque-free gyrostat}.
\item [(ii)] They are not stable with respect to the set of conserved quantities $\mathcal{CQ}$  (see Theorem \ref{stability-conserved-quantities-first-axis}).
\item [(iii)] All of them are Lyapunov stable or all of them are not Lyapunov stable. To prove this affirmation we observe that if $(\vec{M},\vec{\gamma})$ is a solution of \eqref{M} and $a$ is a real number, then $(\vec{M},a\cdot\vec{\gamma})$ is also a solution of \eqref{M}.
\item [(iv)] All o them are Lyapunov stable on the leaf $(H,F,C_1,C_2)(\vec{M},\vec{\gamma})=(H,F,C_1,C_2)(-\mu_1,0,0,\alpha,0,0)$. This leaf is characterized by the relations $\vec{M}=-\vec{\mu}$ and $\vec{\gamma}^2=\alpha^2$. The solution of \eqref{M} which start from $\vec{M}_0=-\vec{\mu},\,\,\vec{\gamma}_0=(\gamma_{10},\gamma_{20},\gamma_{30})$ are $\vec{M}(t)=-\vec{\mu}$ and $\gamma_1(t)=\gamma_{10}$, $\gamma_2(t)=\gamma_{20}\cos(\frac{\mu_1}{I_1}t)-\gamma_{30}\sin(\frac{\mu_1}{I_1}t)$, $\gamma_3(t)=\gamma_{30}\cos(\frac{\mu_1}{I_1}t)
+\gamma_{20}\sin(\frac{\mu_1}{I_1}t)$.
\end{itemize}

For the case of a gyrostatic moment along the second axis of inertia, $\mu_1=0,\,\mu_2>0,\,\mu_3=0$ the Lyapunov stability of the equilibrium states of the form $(0,-\mu_2,0,0,\alpha,0)$ with $\alpha\neq 0$ cannot be decided.

For the case of a gyrostatic moment along the second axis of inertia, $\mu_1=0,\,\mu_2=0,\,\mu_3>0$ the Lyapunov stability of the equilibrium states of the form $(0,0,-\mu_3,0,0,\alpha)$ with $\alpha\neq 0$ cannot be decided.
\end{rem}

\subsection{Partial stability of the equilibrium states}

A practical interest is for the following problems: "is an equilibrium state $\vec{M}$-stable?" or "is it $\vec{\gamma}$-stable?". The problem of $\vec{M}$-stability of an equilibrium state is solved by using Theorem \ref{partial-semi-decoupled-generality} and the results of the paper \cite{comanescu} (see Section \ref{torque-free gyrostat}). In this section we study $\vec{\gamma}$-stability of an equilibrium state of the system \eqref{M}.
We use Theorem 2.3 from \cite{birtea-comanescu} for the system \eqref{M} and the vectorial conserved quantity $(C_1,C_2,F)$. We have that the set
$$\mathcal{M}=\{(\vec{M},\vec{\gamma})\,|\,\text{rank} \left(
                                                   \begin{array}{c}
                                                     \nabla C_1 \\
                                                     \nabla C_2 \\
                                                     \nabla F \\
                                                   \end{array}
                                                 \right)
(\vec{M},\vec{\gamma})=2\}$$
is an invariant set for the dynamics generated by \eqref{M}. By a direct computation we have that
$$\nabla C_1(\vec{M},\vec{\gamma})=(\vec{0},\vec{\gamma}),\,\,\nabla C_2(\vec{M},\vec{\gamma})=(\vec{\gamma},\vec{M}+\vec{\mu}),\,\,\nabla F(\vec{M},\vec{\gamma})=(\vec{M}+\vec{\mu},\vec{0}),$$
and consequently, we have
$$\mathcal{M}=\{(\vec{M},\vec{\gamma})\,|\,\vec{M}+\vec{\mu} \,\,\text{and}\,\,\vec{\gamma}\,\,\text{are linear dependent}\}\setminus \{(-\vec{\mu},\vec{0})\}.$$
The dynamics in the invariant set $\mathcal{M}$ has the following property.

\begin{lem}\label{M-invariant}
If the point $(\vec{M}_0,\vec{\gamma}_0)\in \mathcal{M}$ satisfy $\vec{M}_0+\vec{\mu}=\delta\vec{\gamma}_0$ with $\delta\in \mathbb{R}^*$, then the set $$\mathcal{M}_{(\vec{M}_0,\vec{\gamma}_0)}=\{(\vec{M},\vec{\gamma})\,|\,\vec{M}-\delta\vec{\gamma}+\vec{\mu}=\vec{0}\}$$ is an invariant set of the dynamics generated by \eqref{M}.
\end{lem}

\begin{proof}
We observe that $\vec{\gamma}_0\neq \vec{0}$ and $(\vec{M}_0,\vec{\gamma}_0)\in\mathcal{M}_{(\vec{M}_0,\vec{\gamma}_0)}\subset \mathcal{M}$.
A solution $(\vec{M}(t),\vec{\gamma}(t))$ of \eqref{M} which start from $\mathcal{M}_{(\vec{M}_0,\vec{\gamma}_0)}$ verify $\vec{M}(t)+\vec{\mu}=\delta\vec{\gamma}(t)$.
By the construction of $\mathcal{M}$ we have that there exists a function $h$ such that $\vec{M}(t)+\vec{\mu}=h(t)\vec{\gamma}(t)$.
We deduce
$(\vec{M}(t)+\vec{\mu})\cdot\vec{\gamma}(t)=h(t)\vec{\gamma}^2(t)$
and using the conserved quantities $C_1$ and $C_2$ we obtain
$h(t)=\frac{C_2(\vec{M}(0),\vec{\gamma}(0)}{C_1(\vec{M}(0),\vec{\gamma}(0))}=\delta.$
\end{proof}

\begin{thm} In the case of a gyrostatic moment along the first axis of inertia, $\mu_1>0,\,\mu_2=0,\,\mu_3=0$, we have the following results about the $\vec{\gamma}$-stability of equilibrium states.
\begin{itemize}
\item [a)] $(\vec{M},\vec{\gamma})_0=(0,0,0,\alpha_1,\alpha_2,\alpha_3)$ is $\vec{\gamma}$-stable if and only if $\alpha_2=\alpha_3=0$.
\item [b.1)] $(\vec{M},\vec{\gamma})_{1-2}=(q,0,0,\alpha,0,0)$ is $\vec{\gamma}$-stable if $\alpha=0$ or $q\in (-\infty, -\frac{I_1\mu_1}{I_1-I_2})\cup [-\frac{I_1\mu_1}{I_1-I_3},\infty)$ and $q\neq -\mu_1$.
\item [b.2)] $(\vec{M},\vec{\gamma})_{1-2}=(q,0,0,\alpha,0,0)$ is not $\vec{\gamma}$-stable if  $\alpha\neq 0$ and $q\in [-\frac{I_1\mu_1}{I_1-I_2}, -\frac{I_1\mu_1}{I_1-I_3})$.
\item [c)] $(\vec{M},\vec{\gamma})_4=(\frac{I_1}{I_2-I_1}\mu_1, \beta,0,\frac{\theta\mu_1}{I_2-I_1},\frac{\theta\beta}{I_2},0)$ is $\vec{\gamma}$-stable if and only if $\theta=0$.
\item [d)] $(\vec{M},\vec{\gamma})_5=(\frac{I_1}{I_3-I_1}\mu_1, 0, \beta,\frac{\theta\mu_1}{I_3-I_1},0,
\frac{\theta\beta}{I_3})$ is $\vec{\gamma}$-stable.
\end{itemize}
\end{thm}

\begin{proof}
$a)$ The uniform rotation $(0,0,0)$ is Lyapunov stable for the system \eqref{torque-free}. Using Theorem \ref{partial-semi-decoupled-generality} we deduce that $(\vec{M},\vec{\gamma})_0$ is Lyapunov stable if and only if it is $\vec{\gamma}$-stable. Using Theorem \ref{Lyapunov} we obtain the enounced results.

{\it b.1)} If $\alpha=0$, then we obtain the $\vec{\gamma}$-stability of $(\vec{M},\vec{\gamma})_{1-2}$ by definition and using the conserved quantity $C_1$.
If $q\in (-\infty, -\frac{I_1\mu_1}{I_1-I_2})\cup [-\frac{I_1\mu_1}{I_1-I_3},\infty)$ and $q\neq -\mu_1$, then we obtain the $\vec{\gamma}$-stability of $(\vec{M},\vec{\gamma})_{1-2}$ by using Theorem \ref{Lyapunov}.

{\it b.2)} For the cases in which $\alpha\neq 0$ and $q\in [-\frac{I_1\mu_1}{I_1-I_2}, -\frac{I_1\mu_1}{I_1-I_3})$ the equilibrium state $(\vec{M},\vec{\gamma})_{1-2}$ is not $\vec{M}$-stable. We observe that $(\vec{M},\vec{\gamma})_{1-2}\in \mathcal{M}$ and by using Lemma \ref{M-invariant} we deduce that $\mathcal{M}_{(\vec{M}_{1-2},\vec{\gamma}_{1-2})}$ is an invariant set of the dynamics generated by \eqref{M}. We can apply Theorem \ref{instability} for the function $s(\vec{M},\vec{\gamma})=\vec{M}-\frac{\alpha}{q+\mu_1}\vec{\gamma}+\vec{\mu}$ to deduce the $\vec{\gamma}$-instability of the equilibrium state.

{\it c)} The $\vec{\gamma}$-stability of $(\vec{M},\vec{\gamma})_4$ for $\theta=0$ is proved by using the conserved quantity $C_1$. If $\theta\neq 0$ we have that the equilibrium state $(\vec{M},\vec{\gamma})_4$ is not $\vec{M}$-stable. We use the Lemma \ref{M-invariant} to decide that the equilibrium state $(\vec{M},\vec{\gamma})_4$ is not $\vec{\gamma}$-stable.

{\it d)} The enounced result is the consequence of Theorem \ref{Lyapunov}.

\end{proof}

In the case of a gyrostatic moment along the first axis of inertia, $\mu_1>0,\,\mu_2=0,\,\mu_3=0$, we do not have a result for the $\vec{\gamma}$-stability of the equilibrium states $(-\mu_1,0,0,\alpha,0,0)$ with $\alpha\neq 0$.

\begin{thm} In the case of gyrostatic moment along the second axis of inertia, $\mu_1=0,\,\mu_2>0,\,\mu_3=0$, we have the following results about the equilibrium states.
\begin{itemize}
\item [a)] $(\vec{M},\vec{\gamma})_0=(0,0,0,\alpha_1,\alpha_2,\alpha_3)$ is $\vec{\gamma}$-stable if and only if $\alpha_1=\alpha_3=0$.
\item [b.1)] $(\vec{M},\vec{\gamma})_{1-2}=(0,q,0,0,\alpha,0)$ is $\vec{\gamma}$-stable if $\alpha=0$ or $q\in [-\frac{I_2\mu_2}{I_2-I_3},\frac{I_2\mu_2}{I_1-I_2}]$ and $q\neq -\mu_2$.
\item [b.2)] $(\vec{M},\vec{\gamma})_{1-2}=(0,q,0,0,\alpha,0)$ is not $\vec{\gamma}$-stable if $\alpha\neq 0$ and $q\notin [-\frac{I_2\mu_2}{I_2-I_3},\frac{I_2\mu_2}{I_1-I_2}]$.
\item [c)] $(\vec{M},\vec{\gamma})_3=(\beta, \frac{I_2}{I_1-I_2}\mu_2, 0,\frac{\theta\beta}{I_1},\frac{\theta\mu_2}{I_1-I_2},
0)$ is $\vec{\gamma}$-stable.
\item [d)] $(\vec{M},\vec{\gamma})_5=(0, \frac{I_2}{I_3-I_2}\mu_2, \beta,0,\frac{\theta\mu_2}{I_3-I_2},
\frac{\theta\beta}{I_3})$ is $\vec{\gamma}$-stable.
\end{itemize}
\end{thm}

We do not have a result for the $\vec{\gamma}$-stability of the equilibrium states $(0,-\mu_2,0,0,\alpha,0)$ with $\alpha\neq 0$.

\begin{thm} In the case of gyrostatic moment along the third axis of inertia, $\mu_1=0,\,\mu_2=0,\,\mu_3>0$, we have the following results about the equilibrium states.
\begin{itemize}
\item [a)] $(\vec{M},\vec{\gamma})_0=(0,0,0,\alpha_1,\alpha_2,\alpha_3)$ is $\vec{\gamma}$-stable if and only if $\alpha_1=\alpha_2=0$.
\item [b.1)] $(\vec{M},\vec{\gamma})_{1-2}=(0,0,q,0,0,\alpha)$ is $\vec{\gamma}$-stable if $\alpha=0$ or $q\in (-\infty, \frac{I_3\mu_3}{I_1-I_3}]\cup (\frac{I_3\mu_3}{I_2-I_3},\infty)$ and $q\neq -\mu_3$.
\item [b.2)] $(\vec{M},\vec{\gamma})_{1-2}=(0,0,q,0,0,\alpha)$ is not $\vec{\gamma}$-stable if  $\alpha\neq 0$ and $q\in (\frac{I_3\mu_3}{I_1-I_3}, \frac{I_3\mu_3}{I_2-I_3}]$.
\item [c)] $(\vec{M},\vec{\gamma})_3=(\beta, 0, \frac{I_3}{I_1-I_3}\mu_3,\frac{\theta\beta}{I_1},0,
\frac{\theta\mu_3}{I_1-I_3})$ is $\vec{\gamma}$-stable.
\item [d)] $(\vec{M},\vec{\gamma})_4=(0, \beta, \frac{I_3}{I_2-I_3}\mu_3,0,\frac{\theta\beta}{I_2},
\frac{\theta\mu_3}{I_2-I_3})$ is $\vec{\gamma}$-stable if and only if $\theta=0$.
\end{itemize}
\end{thm}

We do not have a result for the $\vec{\gamma}$-stability of the equilibrium states $(0,0,-\mu_3,0,0,\alpha)$ with $\alpha\neq 0$.
\medskip

{\bf Acknowledgments.} This work was supported by a grant of the Romanian National Authority for
Scientific Research, CNCS UEFISCDI, project number PN-II-RU-TE-2011-3-0006.

\end{document}